\documentclass[3p,lefttitle,sort&compress,fleqn]{elsarticle}
\journal{Physica~D}

\usepackage{amsmath}
\usepackage{booktabs}
\usepackage{multirow}
\usepackage{dcolumn}
\newcolumntype{d}{D{.}{.}{-1}}
\usepackage{chemformula}

\usepackage{amssymb}
\usepackage{amsthm}
\newtheorem{theorem}{Theorem}
\theoremstyle{definition}
\newtheorem{definition}{Definition}

\usepackage{hyperref}
\hypersetup{colorlinks}

\DeclareGraphicsExtensions{.eps}

\usepackage{microtype}
\begin{document}

\begin{frontmatter}

\title{Exact solutions of nonlinear dynamical equations for large-amplitude atomic vibrations in arbitrary monoatomic chains with fixed ends}
\author{George Chechin}
\ead{gchechin@gmail.com}
\author{Denis Ryabov}
\address{Institute of Physics, Southern Federal University, 194 Stachki ave., Rostov-on-Don 344090, Russia}


\begin{abstract}
Intermode interactions in one-dimensional nonlinear periodic structures have been studied by many authors, starting with the classical work by Fermi, Pasta, and Ulam (FPU) in the middle of the last century.
However, symmetry selection rules for the energy transfer between nonlinear vibrational modes of different symmetry, which lead to the possibility of excitation of some bushes of such modes, were not revealed.
Each bush determines an exact solution of nonlinear dynamical equations of the considered system.
The collection of modes of a given bush does not change in time, while there is a continuous energy exchange between these modes.
Bushes of nonlinear normal modes (NNMs) are constructed with the aid of group-theoretical methods and therefore they can exist for the case of large amplitude atomic vibrations and for any type of interatomic interactions.
In most publications, bushes of NNMs or similar dynamical objects in one-dimensional systems are investigated under periodic boundary conditions.
In this paper, we present a detailed study of the bushes of NNMs in monoatomic chains for the case of fixed boundary conditions, which sheds light on a series of new properties of the intermode interactions in such systems.
We prove some theorems that justify a method for constructing bushes of NNMs by continuation of conventional normal modes to the case of large atomic oscillations.
Our study was carried out for FPU chains, for the chains with the Lennard-Jones interatomic potential, as well as for the carbon chains (carbynes) in the framework of the density functional theory.
For one-dimensional bushes (Rosenberg nonlinear normal modes), the amplitude-frequency diagrams are presented and the possibility of their modulational instability is briefly discussed.
We also argue in favor of the fact that our methods and main results are valid for any monoatomic chain.
\end{abstract}


\end{frontmatter}

\section{\label{sec1}Introduction}

There is a conventional approach to the study of atomic vibrations in molecular and crystal structures. It is based on the harmonic approximation, which assumes that the atomic displacements from their equilibrium positions are so small that one can take into account only the quadratic terms in the expansion of the potential energy. In this approximation, the concept of normal modes is introduced, which represent independent dynamical objects. If we excite any normal mode, it continues to exist infinitely long because of this independence.

Taking into account sufficiently small anharmonic terms in the potential energy, one can talk about ``interactions'' between different normal modes.

There are a lot of investigations on the influence of such intermode interactions on the properties of various physical systems (for example, see discussion of the phonon-phonon interactions in crystals~\cite{Reissland1973}).

The simplest system for studying intermode interactions was proposed in Ref.~\cite{FPU1955}. Actually, this work marked the beginning of the development of the modern nonlinear dynamics.

In the above-cited paper, the so-called Fermi-Pasta-Ulam (FPU) chains were introduced. The FPU chain represents a one-dimensional system of masses (hereafter ``particles'') connected by weakly nonlinear springs, which perform longitudinal vibrations along the chain. All masses and springs are identical.

The history of this simple mechanical system takes us to the middle of the last century to Los Alamos, where the big lamp computer MANIAC was used for the purposes of the American atomic project. Since dynamical equations describing the FPU chains were impossible to solve by analytical methods, Enrico Fermi proposed to investigate them numerically with the aid of this computer. In its original formulation, the problem was posed as follows. If we excite the normal mode with the longest wavelength at the initial instant and take into account weak nonlinearities of the interatomic interactions, the energy will gradually pass to other normal modes during the system's time evolution. Then, in accordance with the general concepts of statistical physics, for the case of a sufficiently large number of particles, the final state of the chain should be equipartition of the energy of the initial excitation between all normal modes, which can be considered as the system's degrees of freedom. Fermi was interested in the dependence of the time of such a transition to the state of ``thermodynamic equilibrium'' on the value of the coefficient at the quadratic (FPU-$\alpha$ chain) or cubic (FPU-$\beta$ chain) terms, which are the amendments to the Hooke law.

Instead of the expected energy equipartition between all normal modes, the authors of~\cite{FPU1955} revealed an amazing result described by Zabusky in his brilliant review~\cite{Zabusky1981} by the following words:

``Much to their surprise, the system did not `equilibrate' energy among all the $N$-modes of the system, but rather exhibited long-time near-recurrences and energy sharing only among the lowest modes of the system\ldots and finally returned almost completely to mode~1 after one near-recurrence period~$t_R$.''

The $t_R$ turns out to be a sufficiently long time interval, and the energy of the initial excitation returns to the first mode with an accuracy of 2--3\%. Moreover, this accuracy of the energy return to the initially excited mode does not decrease even after many tens of periods~$t_R$. If we rescale the frequencies of the normal modes to match the audio frequencies perceived by the human ear, then instead of the expected cacophony of sounds, corresponding to the equipartition of energy between all modes, the FPU chains performed a certain melody, in which individual modes were periodically solo on each return cycle (the corresponding music notes were even published!).

This result, which came into conflict with the basic principles of classical statistical physics, was called the FPU paradox. During the decade after the classic work of Fermi, Pasta, and Ulam~\cite{FPU1955}, many attempts have been made to resolve this paradox, but some clarity came only after the appearance of the breakthrough paper by Zabusky and Kruskal~\cite{ZabuskyKruskal1965}. With the aid of the continuum approximation, the authors of this pioneering work made a transition from the system of ordinary differential equations, describing dynamics of the FPU-$\alpha$ chain, to a single differential equation in partial derivatives, which turns out to be the Korteweg de Vries (KdV) equation describing ``solitary waves'' of small amplitude in shallow water. It was in this work that the concept of ``soliton'' appeared for the first time, which became one of the most important concepts of nonlinear physics.

If we make a similar transformation of the dynamical equations for the FPU-$\beta$ chain, we get the so-called modified KdV equation (mKdV). Both of these equations, KdV and mKdV, turned out to be ``completely integrable''---the number of conservation laws (first integrals of motion) corresponding to them is equal to infinity.

A remarkable explicit nonlinear transformation between the solutions of the KdV and mKdV equations was found by Miura in Ref.~\cite{Miura1968}. This transformation has the form of the Riccati equation for the unknown solution of the KdV equation, and by the already known transition to the logarithmic derivative of some new function~$\psi$, it can be reduced to the one-dimensional Schr\"odinger equation. Thus, thanks to Miura, a bridge was thrown from the problem of purely classical physics of solitons to some problems of quantum mechanics. As a result, the so-called ``heroic period'' began in the modern theory of solitons, which is colorfully described by Zabusky in Ref.~\cite{Zabusky1981}.

Thus, the study of the dynamics of such a seemingly elementary mechanical system, which is the FPU model, has led to a very great result of modern physics---to the creation of the theory of solitons, which turned out to be ``ubiquitous'' dynamical objects. Moreover, FPU chains have led to a number of other significant discoveries in nonlinear dynamics. These discoveries are associated with dynamical chaos, discrete breathers, bushes of nonlinear normal modes, etc.\ (see the special issue of the journal Chaos dedicated to the fiftieth anniversary of FPU chains).

In this regard, it is interesting to cite the opinion of Zabusky~\cite{Zabusky2005} in which he comments on two our papers about bushes of nonlinear normal modes in the FPU chains: ``In our new investigation, we will examine the work of Chechin \textit{et al.\ }(2002, 2004) \cite{FPU1, FPU2} to determine if the $n$-curve excitations are some kind of discrete or chaotic breather. Did we discover a new royal path in 1967 and miss an opportunity to explore the richness of a new domain of excitations by stepping off too soon?''

Above we have discussed intermode interactions in the case of the small-amplitude atomic vibrations. In the present paper, we consider intermode interactions in the case of \textit{large-amplitude atomic vibrations} with the aid of the theory of bushes of nonlinear normal modes~\cite{DAN1, PhysD98}.


\section{\label{sec2a}Synopsis of the present paper}

\subsection{The problem to be investigated}

We discuss dynamics of the atomic motion in monoatomic chains with fixed ends and $N$~mobile atoms, under the condition of interatomic interactions between nearest neighbors\footnote{The effect of deviation from this condition is partially discussed in Section~\ref{sec7}.}.

We also make the assumption about the Lyapunov stability of motion~\cite{Lyapunov1992}. Physically, this means that small perturbations of the initial conditions (positions of atoms and their velocities) at $t=0$ lead to small perturbation of the trajectory of motion in $N$-dimensional space for any finite time $t>0$ when solving the Cauchy problem for nonlinear Newton's equations.

\subsection{Basis definitions}

We would like to give some definitions of the dynamical objects in nonlinear systems with discrete symmetry, which are considered in this paper, and briefly discuss their properties.

\subsubsection{Vibrational modes}

--- Conventional \textit{normal modes} (NMs), which are considered in all textbooks on classical mechanics (for example, see Ref.~\cite{LandauLifshitz5}). These modes are obtained in the \textit{harmonic approximation} when one can neglect all terms above the second order of smallness in the decomposition of the potential energy into power series for the case of atomic vibrations with small amplitudes. Normal modes can be obtained by the diagonalization of the force constant matrix whose eigenvalues determine square roots of mode frequencies, while eigenvectors determine their atomic \textit{displacement patterns}. Examples of NMs can be seen in tables in Secs.~\ref{sec4} and~\ref{sec5}.

--- \textit{Nonlinear normal modes by Rosenberg} (Rosenberg modes). They represent some generalization of the conventional normal modes to the case of nonlinear dynamical systems.

Rosenberg mode can be written in the form:
\begin{equation}
{}[a_1,a_2,a_3,\ldots,a_n]\,f(t),
\end{equation}
where $a_i$ are constant coefficients, while $f(t)$ is a certain function of time~$t$. Thus, all degrees of freedom possess the \textit{same time dependence}. The function $f(t)$ can be obtained from the so-called governing differential equation. Note that conventional normal modes also satisfy the definition of Rosenberg modes, where $f(t)$ is a sinusoidal function.

It is essential that Rosenberg modes can exist in the very specific systems, for example, in the systems whose potential energy represents homogeneous function of all its variables.

\subsubsection{Bushes of nonlinear normal modes}

Let us take the displacement pattern of some normal mode and zero velocities of all atoms as initial conditions for numerical solving of the Cauchy problem for nonlinear equations of motion. Then, one usually finds that the displacement pattern of the original normal mode does not retain its symmetry with such a continuation of this mode into the region of nonlinearity.

However, it is proved in the bush theory that for discrete nonlinear dynamical systems, the symmetry of the original normal mode is retained and this leads to appearing of a certain \textit{bush of nonlinear normal modes} (NNMs). In the present paper, we explain this process in detail for the case of monoatomic chains with fixed ends.

A given bush of NNMs can be excited by assigning some initial values to amplitudes of all its modes. However, we will often excite the bush by exciting only one of the system`s vibrational modes. In this case, we use the following definition.

\begin{definition}
The mode, which is excited at the initial instant, is called the \textit{root} mode, while the other modes of the bush that are involved into the vibrational process due to the intermode nonlinear interactions are called the \textit{secondary} modes of this bush.
\end{definition}

\subsection{Symmetry of vibrational states of systems with discrete symmetry}

\subsubsection{Classification of vibrational modes by symmetry}

Nonlinear normal modes can be of different physical nature. In the present paper, only vibrational modes are considered, and, therefore, we use the term ``vibrational mode'' as a synonym of the term ``nonlinear normal mode''.

Let us begin with Wigner's theory and the theory of bushes of nonlinear normal modes. Wigner's theory deals with the classification of conventional (linear) normal modes according to the irreducible representations of the symmetry groups of molecules and crystals in their equilibrium states (more generally, according to the symmetry group of their Hamiltonian). This classification is adequate when considering \textit{small atomic vibrations} in systems with discrete symmetry and is \textit{not correct} in the case of \textit{large} vibrations.

In contrast, the theory of bushes of nonlinear normal modes~\cite{DAN1,PhysD98} was developed specifically for the case of large atomic oscillations. In this case, the classification is based on the symmetry subgroups~$G_j$ of the symmetry group~$G_0$ of the system's equilibrium state (see Sec.~\ref{sec6}). Each bush is associated with several irreducible representations of the group~$G_0$.

Let us clarify what kind of symmetry is used in studying bushes of vibrational modes and how this symmetry is related to the symmetry operations considered in the present paper. The group~$G_0$ is one of the point or space symmetry groups. The elements of point groups are rotations, reflections and improper rotations, while translations, the screw axes and glide planes  are added to them for description of space groups. However, due to the very simple structure of the considered systems, we will use only two symmetry operations---a permutation~$\mathcal{P}$ and an inversion~$\mathcal{I}$, whose meaning requires some explanation.

\subsubsection{Structure of normal modes}

Considering tables of normal modes presented in the present paper (see below Tables~\ref{table1}-\ref{table5}), one can notice that their displacement patterns possess inversion or permutation localized in the center of the chain. This center is located on the central atom for an odd number of the atoms~$N$ in the chain and between two central atoms for even~$N$.

For $N=3$, the first two normal modes possess patterns $[\mu,1,\mu]$ and $[1,0,-1]$. By adding zero displacements of fixed atoms, we can rewrite these patterns in the form $[0,\mu,1,\mu,0]$ and $[0,1,0,-1,0]$. Note that in our tables, the displacement patterns of NMs are represented as columns of the corresponding matrix, but now it is more convenient to write them in the form of rows.

Let us now write the state of the infinite one-dimensional crystal corresponding to the first of the above modes, omitting commas between atomic displacements and depicting bar over the negative components (as it is customary in crystallography)
\begin{equation}\label{eq-ast1}
{}[\ldots|\,0\mu1\mu\,|\,0\bar\mu\bar1\bar\mu\,|\,0\mu1\mu\,|\,0\bar\mu\bar1\bar\mu\,|\,0\ldots],
\end{equation}
and the similar representation for the second mode in the case $N=3$:
\begin{equation}\label{eq-ast2}
{}[\ldots|\,010\bar1\,|\,010\bar1\,|\,010\bar1\,|\,010\bar1\,|\,0\ldots].
\end{equation}

Note that the correct ``continuation'' of our chain (with immobile terminal atoms) requires, in the general case, to consider the periodic boundary conditions for a chain with $2(N+1)$ atoms.

It can be seen from (\ref{eq-ast1}) and (\ref{eq-ast2}) that in these states the crystal structure has basis translation by the value $8a$ and $4a$, correspondingly, where $a$ is the distance between the atoms in the crystal's equilibrium state (or the size of the primitive cell). Thus, during the transition of the crystal structure from the equilibrium state to the states with displaced atoms in accordance with (\ref{eq-ast1}), i.e., upon the symmetry transformation $G_0\to G_j$, the translational symmetry decreased by 8 times, and, therefore, the primitive cell is increased by the same times.

Let's give a few more examples of atomic patterns, taking into account the difference of these patterns for different parity of $N$ and for different parity of the number of normal modes in their ordered set.

For the first mode in the case $N=4$, we have the pattern $[\mu\lambda\lambda\mu]$, which produces the following state of one-dimensional crystal after surrounding by zeros of immobile atoms:
\begin{equation}\label{eq-ast3}
{}[\ldots|\,0\mu\lambda\lambda\mu\,|\,0\bar\mu\bar\lambda\bar\lambda\bar\mu\,|\,0\mu\lambda\lambda\mu\,|\,0\bar\mu\bar\lambda\bar\lambda\bar\mu\,|\,0\ldots].
\end{equation}
Similarly, for the second mode with pattern $[\mu\lambda\bar\lambda\bar\mu]$, we have in the same $N=4$ case:
\begin{equation}\label{eq-ast4}
{}[\ldots|\,0\mu\lambda\bar\lambda\bar\mu\,|\,0\mu\lambda\bar\lambda\bar\mu\,|\,0\mu\lambda\bar\lambda\bar\mu\,|\,0\mu\lambda\bar\lambda\bar\mu\,|\,0\ldots].
\end{equation}

For the $N=5$ case, one can obtain the following results for the first two normal modes:
\begin{gather}
{}[\ldots|\,0\mu\lambda1\lambda\mu\,|\,0\bar\mu\bar\lambda\bar1\bar\lambda\bar\mu\,|\,0\mu\lambda1\lambda\mu\,|\,0\bar\mu\bar\lambda\bar1\bar\lambda\bar\mu\,|\,0\ldots], \label{eq-plus1} \\
[\ldots|\,0\mu\mu0\bar\mu\bar\mu\,|\,0\mu\mu0\bar\mu\bar\mu\,|\,0\mu\mu0\bar\mu\bar\mu\,|\,0\mu\mu0\bar\mu\bar\mu\,|\,0\ldots]. \label{eq-plus2}
\end{gather}

\subsubsection{Symmetry of the displacement patterns}

For clarity, we will describe the symmetry of the normal modes' atomic patterns in the form they are given in our tables, without connection with the corresponding states of the one-dimensional crystal, how it is done in the general bush theory. Therefore, we will use only the operations of simple permutation~($\mathcal{P}$) and inversion~($\mathcal{I}$).

It is easy enough to see both of these symmetry operations in (Tables~\ref{table1}-\ref{table5}), which contain normal mode patterns for different~$N$. All odd modes possess permutation~$\mathcal{P}$ at the center of the chain, while all even modes possess inversion~$\mathcal{I}$ at this center.

However, one can see in Table~\ref{table5} for $N=11$ zeros for some \textit{odd} mode patterns, which must correspond to inversions located at positions of zeros. These elements $\mathcal{I}$ certainly belong to the symmetry group of the atomic patterns, as can be seen from the corresponding states of the one-dimensional crystal.

\subsection{Novelty of the paper}

Studying symmetry properties of normal modes, obtained for the case of small atomic oscillations in the harmonic approximation, makes it possible to construct \textit{bushes of nonlinear normal modes} that can exist at \textit{large} amplitudes of atomic oscillations. These bushes are \textit{exact solutions} of the corresponding nonlinear equations of motion despite large anharmonicity of the problem.

\section{\label{sec2}Theory of the bushes of nonlinear normal modes}

In 1930, Wigner's pioneering paper~\cite{Wigner1930} was published, in which group theoretical methods were used for the first time to analyze and classify normal modes in molecules and crystals. This work has become classical and its results are now presented in all textbooks on spectroscopy of such physical objects.

Wigner showed that normal modes can be classified according to the irreducible representations (irreps) of the symmetry group $G_0$ of the considered system in its equilibrium state (in a more general case, according to the irreps of the group of the system's Hamiltonian). The degeneracy of frequency of any normal mode is determined by the dimension of a certain irrep of the group $G_0$, and the basis vectors of this irrep determine the patterns of atomic displacements corresponding to the given mode.

Let us emphasize that these classical results were obtained for the normal modes, i.e., for the case of small atomic vibrations in the harmonic approximation. In this regard, the following question arises: how will Wigner's results change when passing to the case of large-amplitude atomic vibrations? The theory of bushes of nonlinear normal modes (hereafter the bush theory) gives a complete answer to this question~\cite{DAN1, PhysD98}. Since this theory is based on the group-theoretical methods only, it can be used when taking into account nonlinear terms in the potential energy of any magnitude and of any type for any system with discrete symmetry. The bush theory essentially uses the apparatus of irreducible representations of point and space groups and can be applied for studies of the nonlinear dynamics of complex spatial structures. In particular, it was used in a large series of publications~\cite{IJNLM2003, SF62015, CMS2017, Graphene2017, Diamond2018, PSSb2019, PhSS2019, Ryabov2020} devoted to the study of atomic vibrations in different molecular and crystal structures.

In~\cite{DAN1}, exact selection rules were found for the excitation transfer between vibrational modes of different symmetry. If we use these general selection rules for the case of small oscillations in a system with discrete symmetry, for example, for the FPU chain, then it is possible to determine from which normal mode to which the energy can pass in the process of the time evolution of the system. These selection rules generate significant restrictions on the excitation transfer between different modes, which leads to the possibility to excite some sets of a finite number of modes that can exist in the system without transferring their energy to any other vibrational mode. These dynamical objects have been called \textit{bushes} of nonlinear normal modes. The number of modes entering into a given bush represents its \textit{dimension}, which can be any integer from one to the full dimension of the considered system.

Let's note again that the methods for constructing and studying bushes of modes have been developed for arbitrary nonlinear systems with discrete symmetry, and one-dimensional chain is the simplest object for applying the general bush theory.

Considering interactions between normal modes, we should note several works related to the study of bushes of vibrational modes in FPU chains.

In~\cite{PoggiRuffo1997} for the FPU-$\beta$ chain, some ``closed'' sets of normal modes were revealed that persisted during time evolution: the energy wandered only between these modes, not being transferred to other normal modes. This investigation was based on taking into account the specifics of the FPU-$\beta$ interatomic interactions. It was shown in our papers~\cite{FPU1, FPU2} that the above ``closed'' sets of modes are nothing more than the simplest examples of bushes of nonlinear normal modes.

After the aforementioned paper by Poggi and Ruffo, several papers of other authors were published in which intermode interactions were also studied based on the specifics of interatomic interactions in the FPU chains~\cite{Shinohara2002, Shinohara2003, Yoshimura2004, CafarellaLeo2004}. Only in Refs.~\cite{Rink2001, Rink2003} Rink used some group-theoretical methods for constructing the exact dynamical objects consisting of several vibrational modes, however, this work was performed again precisely for one-dimensional FPU-chains.

In most papers devoted to dynamics of the FPU-chains, periodic boundary conditions were used, which make it easier to construct solutions in the form of bushes of vibrational modes. On the other hand, the phenomenon of the energy return to the initially excited mode was discovered in FPU-chains for the case of \textit{fixed} boundary conditions. Therefore, it is interesting to consider the possibility of the existence of bushes of vibrational modes in this case as well.

In the present work, we show that this problem turns out to be more nontrivial than that for the periodic boundary conditions. Some subtle details of intermode interactions can be revealed by considering the dynamics of monoatomic chains with the fixed boundary conditions. It is essential that we investigate the case of strong anharmonisms corresponding to atomic vibrations with large amplitudes.

\section{\label{sec3}The simplest example with two degrees of freedom}

In this section, we introduce some concepts from the theory of bushes of modes and prove a simple theorem on the intermode interactions in a system with two degrees of freedom.

\subsection{Solution of the Cauchy problem for the chain with two mobile atoms}

We consider a monoatomic chain with 4 atoms, two of which are mobile (internal atoms), while the terminal atoms are fixed.

Since only longitudinal vibrations are considered, this system has two degrees of freedom (and, therefore, two normal modes)---the symmetric mode with the displacement pattern $[a,a]$ and the antisymmetric mode with the pattern $[a,-a]$. Thus, the symmetric mode describes in-phase oscillations of two mobile atoms, while the antisymmetric mode describes their antiphase oscillations.

Hereafter, when speaking about the excitation of a given normal mode, we mean the following procedure.

Using the harmonic approximation, we find the eigenvector of the force-constant matrix that determines the pattern of atomic displacements of this mode and assign some value~$A$ to its amplitude. Then we integrate the nonlinear equations describing the dynamics of the considered system with the aid of some numerical method (as a rule, we use one of the Runge-Kutta methods), assuming the initial velocities of all atoms are equal to zero.

However, let us begin with consideration of the solution of two nonlinear differential equations of the FPU chain in the form of a \textit{power expansion} in a small vicinity of the point $t=0$ for the case of interatomic forces from the FPU-$\alpha$ model
\begin{equation}
F(x) = - x + \alpha x^2.
\end{equation}

Choosing the atomic displacement pattern $[a,a]$ of the symmetric normal mode as the initial conditions with zero velocities of both atoms, we obtain the following power expansions for the time dependence of the displacements $x_1(t)$, $x_2(t)$ of the first and second atoms
\begin{subequations}
\begin{align}
x_1(t) &= a + \left(\frac12 a - \frac12 \alpha a^2\right) t^2 + O(t^4), \label{eq3} \\
x_2(t) &= a + \left(\frac12 a + \frac12 \alpha a^2\right) t^2 + O(t^4). \label{eq4}
\end{align}
\end{subequations}
Similarly, for the case of the antisymmetric mode, which corresponds to a pattern of atomic displacements $[a,-a]$, we have
\begin{subequations}
\begin{align}
x_1(t) &= \phantom{+}a + \left(\frac32 a + \frac32 \alpha a^2\right) t^2 + O(t^4),  \label{eq5} \\
x_2(t) &= -a - \left(\frac32 a + \frac32 \alpha a^2\right) t^2 + O(t^4).\label{eq6}
\end{align}
\end{subequations}
Note that all these expansions are in even powers of the variable~$t$, due to the symmetry of the Newton equations with respect to the time reversal.

In these equations, the coefficients at the quadratic terms in the variable~$t$ are directly related to the accelerations of the first and second atoms, and therefore, by virtue of Newton's second law, they determine the forces acting on our two mobile atoms.

Comparison of equations (\ref{eq3}) and (\ref{eq4}) shows that for the symmetric mode the accelerations of both atoms are different, and therefore, the forces acting on these atoms are also different, even for arbitrarily small times near the initial instant $t=0$. This is also easy to understand from a direct examination of the forces acting on these atoms (see Theorem~\ref{theorem1}).

\subsection{Symmetry analysis of intermode interactions}

The pattern $[a,a]$ is invariant under \textit{permutation} of its components (operation~$\mathcal{P}$ relative to the center of this short chain, which is located between two mobile atoms), while the pattern $[a,-a]$ is invariant under \textit{inversion} located at the same center (operation~$\mathcal{I}$).

These normal modes are independent of each other in the harmonic approximation and their time dependencies are described by sinusoidal functions with different frequencies.

If we take into account some small anharmonic terms in the potential energy, a weak interaction between the above discussed modes arises and we can prove a simple theorem:

\begin{theorem}\label{theorem1}
The excitation (energy) cannot transfer from the antisymmetric mode to the symmetric mode, while the excitation transfer from the symmetric mode to the antisymmetric mode can occur already at arbitrarily small times.
\end{theorem}
\begin{proof}
1. Let us excite the antisymmetric mode with the displacement pattern $[a,-a]$ at the initial instant $t=0$. We can assume without loss of generality that $a>0$. Then both mobile atoms are shifted from their equilibrium positions to the center of the chain by equal distances~$a$ (see Fig.~\ref{fig3}).

\begin{figure}
  \centering
  \includegraphics[width=.5\linewidth]{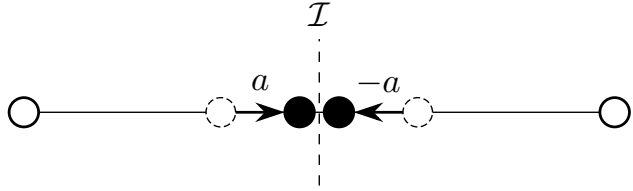}
  \caption{\label{fig3} The antisymmetric mode with the displacement pattern $[a,-a]$.}
\end{figure}

Consequently, the forces acting on them from the terminal atoms are the same in magnitude and opposite in direction. The same can be said about the forces acting between mobile atoms due to the Newton's third law.

Then, it follows from the uniqueness of the Cauchy problem for Newton's nonlinear equations that antisymmetry of the mobile atoms structure survives for $t>0$.

This means that the antisymmetric mode can exist for an arbitrarily long time and for sufficiently large values of the parameter~$a$ of its displacement pattern $[a,-a]$. In turn, this means that the excitation (energy) from it is not transferred to the symmetric mode, because otherwise, during evolution, the initial pattern $[a,-a]$ would turn into a general pattern $[a,b]$ with $a\not\equiv b$.

2. Let us now consider the case when only symmetric mode with pattern $[a,a]$ was excited at the initial instant $t=0$ for solving the Cauchy problem. If $a>0$, then both mobile atoms are displaced by equal distances to the right, i.e., the first mobile atom moves away from the left fixed atom of the chain, while the second atom approaches the right end of the chain.

The interaction forces between these two atoms are equal in magnitude according to Newton's third law, and, therefore, the total forces acting on the first and second mobile atoms are obviously different\footnote{Except of the specific case of an even interatomic potential.}. This fact leads to changing the initial pattern $[a,a]$ of the symmetric mode to the form $[a,b]$, where $a\not\equiv b$.

Thus, the excitation transfer from the symmetric mode to the antisymmetric mode occurs even at arbitrarily small time~$t$.

Indeed, we can introduce in two-dimensional space, corresponding to the vector $\vec r=[a,b]$, a basis of two orthonormal vectors $\frac{1}{\sqrt2}[1,1]$ and $\frac{1}{\sqrt2}[1,-1]$. Expanding $\vec r$ in this basis, we get $\vec r=A_1\frac{1}{\sqrt2}[1,1]+A_2\frac{1}{\sqrt2}[1,-1]$. Thus, vector $\vec r$ can be considered as a \textit{linear combination of two modes}---symmetric and antisymmetric. In other words we can say that, being excited at the initial instant, the symmetric mode involves the antisymmetric mode into the vibrational process even at arbitrarily small times.
\end{proof}

Note that in this very simple case, we encounter for the first time a~certain \textit{two-dimensional bush} consisting of symmetric and antisymmetric modes.
In this case, the symmetric mode plays the role of the \textit{root} mode, and the antisymmetric mode is the \textit{secondary} mode, which is involved into the vibrational process due to nonlinear terms in dynamical equations.

Let us again consider the solution of the FPU nonlinear equations near $t=0$ represented by Eqs.~(\ref{eq3}--\ref{eq6}).

Since positions of mobile atoms, determined by the presence of the inversion center between them, is preserved in time, it is obvious that the coefficients of series~(\ref{eq5}) will be strictly equal in magnitude and opposite in sign to those of series~(\ref{eq6}). In other words, the inversion symmetry of the antisymmetric mode is conserved in time due to the uniqueness of the solution of the Cauchy problem for the system of ordinary differential equations, if we consider its behavior in the small vicinity of the initial instant.

Since inversion symmetry is conserved during the time evolution, the antisymmetric mode can exist as an exact dynamical object for an arbitrarily long time. In fact, it is a~certain Rosenberg mode (see~\cite{Chechin2020review}).

The main result of our consideration of the case $N=2$ is the statement that the symmetric mode involves the antisymmetric mode into the vibrational process, while the antisymmetric mode does not involve the symmetric mode. This conclusion is based on symmetry-related arguments, due to which it should take place for any monoatomic chains and it is valid for any amplitude of the considered modes up to the loss of stability of the antisymmetric mode (see below).

\section{\label{sec4}Symmetry of normal modes}

The normal modes can be found as a result of solving the eigenproblem for the force constant matrix for the considered system, i.e., the matrix of second order partial derivatives of the potential energy in the harmonic approximation. Each eigenvector determines a pattern of atomic displacements, which is characterized by a certain symmetry group. In the case of the monoatomic chain with fixed ends and interactions only between the nearest neighbors, one can obtain an analytical formula for the eigenvectors and eigenvalues of the force constant matrix~$W$. This formula was used already in the first work devoted to the FPU chains~\cite{FPU1955}.

Let $k=1..N$ be the index of the column of the force constant matrix, which determines the number of normal modes, and $n=1..N$ is the index of its row, which determines the position of the atom in the chain. Then we have
\begin{align}
\omega_k &= 2\omega_0\sin\frac{\pi k}{2(N+1)}, \label{eq1} \\
X_{k,n}  &= A_k\sqrt{\frac{2}{N+1}}\sin\frac{\pi kn}{N+1}. \label{eq2}
\end{align}
Here, $\omega_k$ is the cyclic frequency of the $k$th mode, $A_k$ is its amplitude, $X_{k,n}$ is the displacement of the $n$th atom from its equilibrium position in the pattern of the normal mode with the number~$k$, and the frequency $\omega_0$ is determined by the interparticle potential [$\omega_0=1$ for FPU models and $3\sqrt[6]{4B_{LJ}/A_{LJ}^4}$ for the case of the Lennard-Jones potential~(\ref{eq10})].

Using Eq.~(\ref{eq2}), one can analyze the symmetry of the displacement patterns of the normal modes for arbitrary~$N$. Let us consider this problem in more detail.

In the second row of Table~\ref{table1}, we show in ascending order the frequencies of normal modes for the case $N=7$. The corresponding eigenvectors of the matrix~$W$ are shown below in columns of the $7\times 7$ matrix. Similar information for the case of the chain with $N=8$ atoms is shown in Table~\ref{table2}.

\begin{table}
  \centering
  \begin{tabular}{rddddddd}
    \toprule
    Mode & 1     & 2      & 3      & 4      & 5      & 6      & 7      \\
    \cmidrule(l){2-8}
    Frequency & 1.475 & 2.893  & 4.200  & 5.345  & 6.286  & 6.984  & 7.414  \\
    \cmidrule(l){2-8}
    \multirow{7}{*}{\rotatebox{90}{Pattern}} & %
     0.191 & 0.354  & 0.462  & 0.500  & 0.462  & 0.354  & 0.191  \\
    &0.354 & 0.500  & 0.354  & {\bf 0}& -0.354 & -0.500 & -0.354 \\
    &0.462 & 0.354  & -0.191 & -0.500 & -0.191 & 0.354  & 0.462  \\
    &0.500 & {\bf 0}& -0.500 & {\bf 0}& 0.500  & {\bf 0}& -0.500 \\
    &0.462 & -0.354 & -0.191 & 0.500  & -0.191 & -0.354 & 0.462  \\
    &0.353 & -0.500 & 0.354  & {\bf 0}& -0.354 & 0.500  & -0.354 \\
    &0.191 & -0.354 & 0.462  & -0.500 & 0.462  & -0.354 & 0.191  \\
    \bottomrule
  \end{tabular}
  \caption{\label{table1}Patterns of atomic displacements of normal modes for the chain with $N=7$ atoms.}
\end{table}

\begin{table}
  \centering
  \begin{tabular}{D{.}{.}{4}ddddddd}
    \toprule
    1.313 & 2.586  & 3.780  & 4.859  & 5.791  & 6.547  & 7.104  & 7.445  \\
    \midrule
    0.161 & 0.303  & 0.408  & 0.464  & 0.464  & 0.408  & 0.303  & 0.161  \\
    0.303 & 0.464  & 0.408  & 0.161  & -0.161 & -0.408 & -0.464 & -0.303 \\
    0.408 & 0.408  & {\bf 0}& -0.408 & -0.408 & {\bf 0}& 0.408  & 0.408  \\
    0.464 & 0.161  & -0.408 & -0.303 & 0.303  & 0.408  & -0.161 & -0.464 \\
    0.464 & -0.161 & -0.408 & 0.303  & 0.303  & -0.408 & -0.161 & 0.464  \\
    0.408 & -0.408 & {\bf 0}& 0.408  & -0.408 & {\bf 0}& 0.408  & -0.408 \\
    0.303 & -0.464 & 0.408  & -0.161 & -0.161 & 0.408  & -0.464 & 0.303  \\
    0.161 & -0.303 & 0.408  & -0.464 & 0.464  & -0.408 & 0.303  & -0.161 \\
    \bottomrule
  \end{tabular}
  \caption{\label{table2}Patterns of atomic displacements of normal modes for the chain with $N=8$ atoms.}
\end{table}

Let us pay attention to the fact that the character of symmetry of atomic patterns of normal modes differs for the chains with odd and even number of atoms.

One can reveal the following two types of symmetry elements of normal modes: permutation~$\mathcal{P}$ of vector components for ``conjugate'' atoms, i.e., atoms located at an equal distance from the center of the permutation, and inversion~$\mathcal{I}$, which is a permutation of the components of conjugated atoms with the change in their sign. Both of these symmetry elements can be centered either at an atom or in the midpoint between two atoms.

For example, the first column in Table~\ref{table1} is invariant under the action of $\mathcal{P}$ located at the fourth atom of the chain, while the first column from Table~\ref{table2} is invariant under the action of~$\mathcal{P}$, which is centered between the fourth and fifth atoms of the chain. In both these cases, operation~$\mathcal{P}$ is localized in the center of the chain.

The second column in Table~\ref{table1} is invariant under the action of inversion centered at the fourth atom, while the second column in Table~\ref{table2} is invariant with respect to the inversion localized between the 4th and 5th atoms of the chain. Note that if the inversion is localized at a certain atom, then its displacement from the equilibrium position is zero (this atom is immobile during time evolution).

It can be seen from Table~\ref{table1} that all its odd columns are invariant under the operation~$\mathcal{P}$, while all even columns are invariant under the operation~$\mathcal{I}$ localized at the center of the chain. Moreover, the fourth column of this table, which determines the pattern of atomic displacements in the fourth normal mode, has three immobile atoms: in addition to the inversion in the center of the chain, it has two more inversions centered at the second and sixth atoms.

In all odd columns of Table~\ref{table2}, there is the operation~$\mathcal{P}$ localized at the center of the chain (between the fourth and fifth atoms), and in all even columns, there is the operation~$\mathcal{I}$ that is also localized at the center of the chain. In addition, in the third and sixth columns of Table~\ref{table2}, there are operations~$\mathcal{I}$, localized at the third and fifth atoms, whereby these atoms are immobile.

Using the displacement patterns shown in the above tables and zero velocities of all the atoms as initial conditions, we can use numerical methods in order to obtain an approximate continuation of a given normal mode to the case of large amplitudes.

One essential question arises in connection with this process: whether the symmetry of the initial conditions is inherited by dynamical regimes appearing as a result of the integration of nonlinear differential equations? The point is that the existence of bushes of NNMs as some exact solutions of nonlinear equations is dictated by the conservation of the symmetry of the dynamical regimes during their time evolution~\cite{DAN1, PhysD98, Chechin2020review}. We show in the next section that simple permutation symmetry is not inherited, in contrast to the inversion symmetry, which is preserved during the time evolution of the initially excited dynamical regime.

\section{\label{sec5}Conservation of the symmetry of dynamical regimes}

\subsection{Solution of the Cauchy problem for dynamics of the Lennard-Jones chains}

In this section, we consider the dynamics of monoatomic chains with the Lennard-Jones interparticle potential~\cite{LennardJones1931} (L-J chains)
\begin{equation}\label{eq10}
  u(r) = \frac{A_{LJ}}{r^{12}} - \frac{B_{LJ}}{r^6}.
\end{equation}

We assume $A_{LJ}=B_{LJ}=1$, which is not a loss of generality since this can be achieved with the aid of the appropriate scaling of temporal and spatial variables.

The analytical form of power expansions for solutions of the L-J dynamical equations are very complex and, therefore, we present the results obtained by the Runge-Kutta numerical method. In order to find such solutions, one must set some value of the amplitude~$A$ of the initially excited mode. All results presented below have been obtained for the case $A=0.1$.

Let us consider the L-J chain with three mobile atoms
($N=\nolinebreak 3$).
There are three normal modes in this case, whose symmetry elements are localized at the second atom. The modes 1 and 3 possess the permutation symmetry~($\mathcal{P}$), while the symmetry of the second mode is the inversion~($\mathcal{I}$).

The time dependence of the atomic displacements for the first mode ($k=1$) has the following form
\begin{subequations}
\begin{align}
x_1(t) &= 0.0500 -0.398 t^2+O(t^4), \\
x_2(t) &= 0.0707 -0.302 t^2+O(t^4), \\
x_3(t) &= 0.0500 -0.104 t^2+O(t^4).
\end{align}
\end{subequations}

These equations show that accelerations of all the atoms are opposite in sign to their displacements and that the accelerations of the conjugated atoms (relative to the operation~$\mathcal{P}$) turn out to be different, i.e., the $\mathcal{P}$-symmetry is lost already at arbitrarily small times.

The time dependence of atomic displacements for the second mode ($k=2$) has the form
\begin{subequations}
\begin{align}
x_1(t) &= \phantom{+}0.0707 -1.273 t^2+O(t^4), \\
x_2(t) &= \phantom{+}0      +0.0   t^2+O(t^4), \\
x_3(t) &= -0.0707           -1.273 t^2+O(t^4).
\end{align}
\end{subequations}

It follows from these equations that the $\mathcal{I}$-symmetry of the initial conditions is preserved in a small vicinity of $t=0$, which means that it is also preserved globally. Indeed, Newton's dynamical equations do not contain anything else but the accelerations of the atoms and the forces acting on them. Therefore, the motion is uniquely determined by its behavior in the vicinity of $t=0$. In other words, it can be seen from the displacement pattern of mode~2 that the initial arrangement of all the atoms is such that the set of forces acting on them also has $\mathcal{I}$-type of symmetry, which means that the pattern of atomic displacements has this symmetry at any other instant.

The numerical integration of the dynamical equations for the L-J chain confirms the above conclusions: the initial permutation symmetry of modes 1 and 3 is broken even at arbitrarily small times, while the inversion symmetry of the second mode persists during the entire time of our numerical experiments. In other words, the excitation is not transmitted from the second mode to other modes. Since the inversion in this mode is localized on the second atom, it remains immobile.

Let us now consider the case of the L-J chain with four mobile atoms. The displacement patterns of the corresponding four normal modes are shown in Table~\ref{table3}. It follows from this table that for modes 1 and 3 there is permutation~$\mathcal{P}$ localized between the second and third atoms, while for modes 2 and 4 there is inversion~$\mathcal{I}$ also localized between these atoms.

\begin{table}
  \centering
  \begin{tabular}{D{.}{.}{4}ddd}
    \toprule
    2.336 & 4.443  & 6.116  & 7.190  \\
    \midrule
    0.372 & 0.602  & 0.602  & 0.372  \\
    0.602 & 0.372  & -0.372 & -0.602 \\
    0.602 & -0.372 & -0.372 & 0.602  \\
    0.372 & -0.602 & 0.602  & -0.372 \\
    \bottomrule
  \end{tabular}
  \caption{\label{table3}Patterns of atomic displacements for normal modes of L-J chain with 4 mobile atoms.}
\end{table}

As before, we consider the solution of corresponding dynamical equations in a small vicinity of $t=0$, choosing the patterns of atomic displacements of normal modes and zero velocities of all the atoms as the initial conditions.

The following expansions for the time evolution of atomic displacements for mode~2 are obtained
\begin{subequations}
\begin{align}
x_1(t) &= \phantom{+}0.06015 -0.90225 t^2+O(t^4), \\
x_2(t) &= \phantom{+}0.03717 -0.13971 t^2+O(t^4), \\
x_3(t) &= -0.03717 +0.13971 t^2+O(t^4), \\
x_4(t) &= -0.06015 +0.90225 t^2+O(t^4).
\end{align}
\end{subequations}
Similarly, for mode~4 the following expansions are obtained
\begin{subequations}
\begin{align}
x_1(t) &= \phantom{+}0.03717 -0.67213 t^2+O(t^4), \\
x_2(t) &= -0.06015 +3.17220 t^2+O(t^4), \\
x_3(t) &= \phantom{+}0.06015 -3.17220 t^2+O(t^4), \\
x_4(t) &= -0.03717 +0.67213 t^2+O(t^4).
\end{align}
\end{subequations}

It can be seen from these equations that for both modes 2 and 4 the inversion symmetry~$\mathcal{I}$ of initial conditions is preserved in the solution near $t=0$.
On the other hand, for mode~3 we have the following expansions
\begin{subequations}
\begin{align}
x_1(t) &= \phantom{+}0.06015 -1.06280 t^2+O(t^4), \\
x_2(t) &= -0.03717 +0.29330 t^2+O(t^4), \\
x_3(t) &= -0.03717 +1.82352 t^2+O(t^4), \\
x_4(t) &= \phantom{+}0.06015 -2.07288 t^2+O(t^4).
\end{align}
\end{subequations}

Thus, the permutation symmetry~$\mathcal{P}$ located at the center of the chain between the second and third atoms, which took place in the pattern of initial displacements, is lost at arbitrarily small times since the accelerations of the atoms do not possess such symmetry. The same situation of the permutation symmetry loss occurs for mode~1.

\subsection{Symmetry of intermode interactions}

\begin{definition}
The vibrational mode in the monoatomic chain with fixed ends will be called \textit{inversion mode} if its displacement pattern possesses inversion (operation~$\mathcal{I}$) with respect to the center of the chain, while the mode with the permutation at the center of the chain will be called \textit{permutation mode}.
\end{definition}

The displacement pattern of inversion mode can be written as
\begin{equation}
{}[\ldots a, b, c \,|\, -c, -b, -a, \ldots] \quad\text{or}\quad [\ldots a, b \,|\, 0 \,|\, -b, -a, \ldots]
\end{equation}
(depending on the parity of $N$), while that of the permutation mode can be written in the form
\begin{equation}
{}[\ldots a, b, c \,|\, c, b, a, \ldots] \quad\text{or}\quad [\ldots a, b \,|\, c \,|\, b, a, \ldots],
\end{equation}
where $a, b, c, \ldots$ are arbitrary parameters.

We will also say that the displacement pattern of an inversion mode \textit{has inversion structure}, while that of a permutation mode has \textit{permutation structure}.

If we sort the list of normal modes by increasing of their frequencies, then we can find that all even modes are inversion modes, while all odd ones represent permutation modes. In the pattern of an even mode, there is the inversion at the center of the chain, which is located either at an atom (for odd~$N$) or between atoms (for even~$N$). These properties can be easily revealed with the aid of formulas (\ref{eq1}, \ref{eq2}) that are valid for the case of monoatomic chains with fixed ends and interactions between nearest neighbors.

For example, the center of the chain with $N=19$ mobile atoms corresponds to the position of the 10th atom and we find from Eq.~(\ref{eq2}) that the displacement of this atom is equal to $\sin(\pi 10k/20)=\sin(\frac{\pi}{2}k)$. For all \textit{even} modes this displacement is zero, and the inversion structure of the displacement pattern appears in the vicinity of the center of the chain that is also obvious from Eq.~(\ref{eq2}).

Moreover, we can reveal with the aid of Eq.~(\ref{eq2}) the presence of zeros on the positions of atoms 5 and 15, as well as inversion structures in their vicinity.

Let us recall that we use the term \textit{conjugate atoms} for atoms located at the same distance from the center of the chain.

\begin{theorem}
Any inversion mode in the monoatomic chain with arbitrary interactions between nearest neighbors and fixed ends exists for an infinitely long time without transferring excitation to permutation modes\footnote{The excitation transfer from a given inversion mode to other modes of this type is not prohibited and, moreover, such a transfer of excitation to other inversion modes takes place during the formation of the bush with a dimension greater than unity.}. Any permutation mode can transfer excitation to inversion modes even at arbitrarily small times.
\end{theorem}

\begin{proof}
In any inversion mode, the displacements of conjugated atoms are equal in magnitude and opposite in direction. Such a displacement structure dictates the same symmetry of the \textit{forces} acting on conjugated atoms, which means that \textit{accelerations} of conjugated atoms also have the same symmetry.

On the other hand, the process of numerically solving nonlinear dynamical equations begins with specifying the initial conditions at $t=0$ determined by the pattern of displacements of a given mode multiplied by a certain amplitude~$A$.

Then the initial positions of conjugated atoms in this case possess inversion structure. On the other hand, initial velocities of all the atoms are equal to zero and, therefore, we can consider that these velocities also possess the inversion structure.

Then it follows from Newton's second law that at any time~$t$ the atomic configuration retains the inversion structure.

Thus, we can say that inversion preserves during the atomic movement of the considered chain. Note that this statement is true under condition that the dynamical process \textit{does not lose the Lyapunov stability} (see below).

Let us now consider the second part of Theorem 2.

In the displacement pattern $[a, \ldots, a]$ of the permutation mode, the first mobile atom moves away by~$a$ (for $a>0$) from the left fixed atom of the chain, while the $N$th atom approaches by the same amount~$a$ the right fixed atom.

Therefore, already at $t=0$, the forces acting on the above mobile atoms are \textit{different}\footnote{Again, except of the specific case of an even interatomic potential.}. As a result, the accelerations of these atoms are also different. This leads to a \textit{change in the symmetry} of the displacement pattern, since it turns into the pattern $[a, \ldots, b]$ with $a\not\equiv b$.

Thus, the initial permutation symmetry of the displacement pattern is lost even at arbitrarily small evolution times~$t$.
\end{proof}

\section{\label{sec6}Bushes of nonlinear normal modes in the L-J chains with fixed ends}

\subsection{Excitation transfer between vibrational modes of different symmetry}

We continue to discuss dynamics of the atomic motion in monoatomic chains with $N$ mobile atoms and fixed ends, under the condition of interatomic interactions only between nearest neighbors. Let us give some definitions and prove a theorem, which can help us to find different bushes of vibrational modes.

\begin{definition}\label{def1}
The set of atomic positions, which correspond to zeros in the pattern of a given vibrational mode, will be called the \textit{Z-structure} of this mode. This structure is not changed in time---all its atoms remain motionless during time evolution.
\end{definition}

\begin{definition}\label{def2}
If the Z-structure $Z_A$ of the mode~$A$ contains all the atomic positions presented in the Z-structure $Z_B$ of the mode~$B$, and some additional atomic positions that are not contained in $Z_A$ (i.e., $Z_A$ is a subset of the set $Z_B$), we will say that $Z_B$ \textit{covers} the structure $Z_A$ of the mode~$A$.
\end{definition}

\begin{theorem}\label{theorem3}
Any mode with the structure $Z_0$ involves into the vibrational process those and only those vibrational modes that have either the identical Z-structures, or Z-structures, which cover the $Z_0$-structure of the original mode.
\end{theorem}

This theorem does not require any special proof, because it is a direct consequence of one of the main statements, which are proved in the theory of the bushes of nonlinear normal modes~\cite{PhysD98}. Indeed, it is known from this theory that excitation (energy) from any mode can be transferred only to the modes whose symmetry groups are equal or higher than that of the given mode, but not to the modes with lower symmetry groups.

The fact that Theorem~\ref{theorem3} follows from the above statement of the bush theory is easy to see by considering the states of the one-dimensional crystal corresponding to Z-structures, as was described in Sec.~\ref{sec2a} of the present paper.

\subsection{Bushes of vibrational modes for different~$N$}

Let us continue to discuss the L-J chain with four atoms. We already saw that the inversion symmetry~$\mathcal{I}$ that took place in the initial atomic pattern for modes 2 and 4 is preserved during the time evolution. What happens if only one of these modes is excited at the initial instant?

If we excite mode~2 (root mode) and carry out the numerical integration of nonlinear Newton's equations, we find that this mode involves into the vibrational process the previously unexcited mode~4 (secondary mode) and no other mode. The energy is not transmitted to modes 1 and 3, i.e., they remain ``sleeping'' modes. Let's consider this issue in more detail.

Displacement patterns shown in the columns of Table~\ref{table3} for the case $N=4$ represent the orthonormal set of eigenvectors of the force constant matrix, which forms a certain basis of the four-dimensional space of all atomic displacements. If we stop the numerical integration at some arbitrary moment~$t$, then the vector $\vec X(t)$ representing the set of atomic displacements can be expanded in this basis. If the projection of the vector $\vec X(t)$ onto some basis vector is equal to zero, this means that the corresponding mode is not excited due to the interaction with the initially excited root mode. The following result has been obtained for the initially excited mode~2 at $t=2.0$:
\begin{equation}
\vec X=[0,-0.09398,0,0.00713].
\end{equation}

This means that only two modes turned out to be excited at the moment ${t=2.0}$---the root mode~2, which was excited at the initial instant, and the secondary mode~4, which was excited through its interaction with the root mode during numerical integration of the equations of motion.

Note that in the considered case the atomic displacements, corresponding to the secondary mode~4, is about an order of magnitude smaller than the amplitude of the root mode~2.

In the considered case, we come across the concept of the bushes of vibrational modes. We obtain the two-dimensional bush, consisting of two nonlinear modes 2 and 4. If we excite mode~4 at the initial instant (it will be the root mode in such a case), then during time evolution the mode~2 is involved into the vibrational process (and only this mode), resulting in the same two-dimensional bush, consisting of the modes 2 and 4.

The time evolution of this bush is shown in Fig.~\ref{fig1}. It can be seen from this figure that both modes of the bush exchange energy all the time. It is important to note that both normal modes, from which two nonlinear modes were constructed, possess incommensurate frequencies (4.443, 7.190). Thus, the two-dimensional bush describes a quasiperiodic motion---there are two base frequencies in its Fourier spectrum, as well as their different integer linear combinations, which appear because of the nonlinearity of the dynamical process.

\begin{figure}
  \centering
  \includegraphics[width=.7\linewidth]{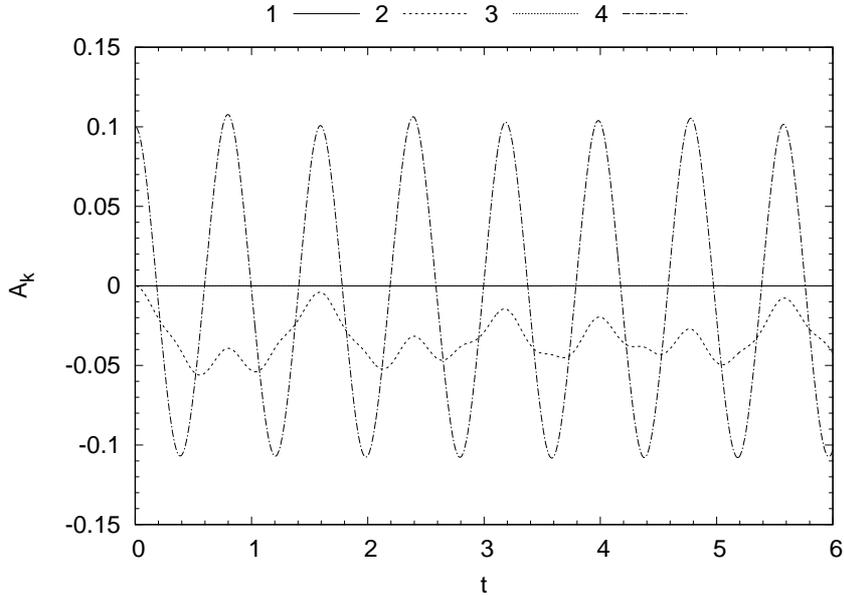}
  \caption{\label{fig1} Time evolution of the two-dimensional bush in the L-J chain with four mobile atoms.}
\end{figure}

It has been proved in the bush theory~\cite{PhysD98}, that the symmetry of the secondary modes should not be lower than that of the root mode, i.e., it must be the same as that of the root mode or higher. In our case, both modes 2 and 4 possess the same symmetry, which is characterized by the inversion~$\mathcal{I}$ located between atoms 2 and 3 of the chain. Because of this, both modes can equally act as root mode for the considered two-dimensional bush.

Let us consider now the case of the L-J chain with $N=11$ atoms. Atomic patterns of all normal modes for this case are given in Table~\ref{table5}.

\begin{table}
  \centering
  \small
  \setlength{\tabcolsep}{0pt}
  \begin{tabular}{D{.}{.}{4}dddddddddd}
    \toprule
    0.987& 1.957 & 2.893 & 3.780 & 4.602 & 5.345 & 5.997 & 6.547 & 6.984 & 7.302 & 7.495  \\
    \midrule
    0.106& 0.204 & 0.289 & 0.354 & 0.394 & 0.408 & 0.394 & 0.354 & 0.289 & 0.204 & 0.106  \\
    0.204& 0.354 & 0.408 & 0.354 & 0.204 &{\bf 0}& -0.204& -0.354& -0.408& -0.354& -0.204 \\
    0.289& 0.408 & 0.289 &{\bf 0}& -0.289& -0.408& -0.289&{\bf 0}& 0.289 & 0.408 & 0.289  \\
    0.354& 0.354 &{\bf 0}& -0.354& -0.354&{\bf 0}& 0.354 & 0.354 &{\bf 0}& -0.354& -0.354 \\
    0.394& 0.204 & -0.289& -0.354& 0.106 & 0.408 & 0.106 & -0.354& -0.289& 0.204 & 0.394  \\
    0.408&{\bf 0}& -0.408&{\bf 0}& 0.408 &{\bf 0}& -0.408&{\bf 0}& 0.408 &{\bf 0}& -0.408 \\
    0.394& -0.204& -0.289& 0.354 & 0.106 & -0.408& 0.106 & 0.354 & -0.289& -0.204& 0.394  \\
    0.354& -0.354&{\bf 0}& 0.354 & -0.354&{\bf 0}& 0.354 & -0.354&{\bf 0}& 0.354 & -0.354 \\
    0.289& -0.408& 0.289 &{\bf 0}& -0.289& 0.408 & -0.289&{\bf 0}& 0.289 & -0.408& 0.289  \\
    0.204& -0.354& 0.408 & -0.354& 0.204 &{\bf 0}& -0.204& 0.354 & -0.408& 0.354 & -0.204 \\
    0.106& -0.204& 0.289 & -0.354& 0.394 & -0.408& 0.394 & -0.354& 0.289 & -0.204& 0.106  \\
    \bottomrule
  \end{tabular}
  \caption{\label{table5}Patterns of atomic displacements of normal modes for the chain with $N=11$ atoms.}
\end{table}

This table shows that on the sixth atom (it is the center of the chain) the patterns of all even modes possess inversion~$\mathcal{I}$, due to which this atom is immobile. All odd modes possess permutation~$\mathcal{P}$ at the center of the chain.

Some modes possess inversions at the other atoms of the chain as well, due to which these atoms are also immobile.

Different symmetries of the displacement patterns of different modes lead to some \textit{selection rules} for the excitation transfer between these modes. It is important to emphasize that these selection rules, being deduced from symmetry-related properties, are valid for vibrations with arbitrary amplitudes and for monoatomic chains with fixed ends for any type of interatomic interactions. Let's consider these issues in more detail.

The maximal number of immobile atoms in the displacement patterns from Table~\ref{table5} corresponds to mode~6 (five zeros in the sixth column). Therefore, it is the mode of higher symmetry among all other modes. According to the bush theory, the excitation (energy) from this mode cannot pass to other modes that possess a lower symmetry. Thus, mode~6 can exist in the system indefinitely long without excitation of any other modes. It represents the one-dimensional bush (nonlinear normal mode by Rosenberg~\cite{Rosenberg1962, Rosenberg1966}).

The amplitude-frequency diagram of mode~6 is presented in Fig.~\ref{fig2}, from which the \textit{hard type} of nonlinearity can be seen (the frequency increases with increasing amplitude).

\begin{figure}
  \centering
  \includegraphics[width=.7\linewidth]{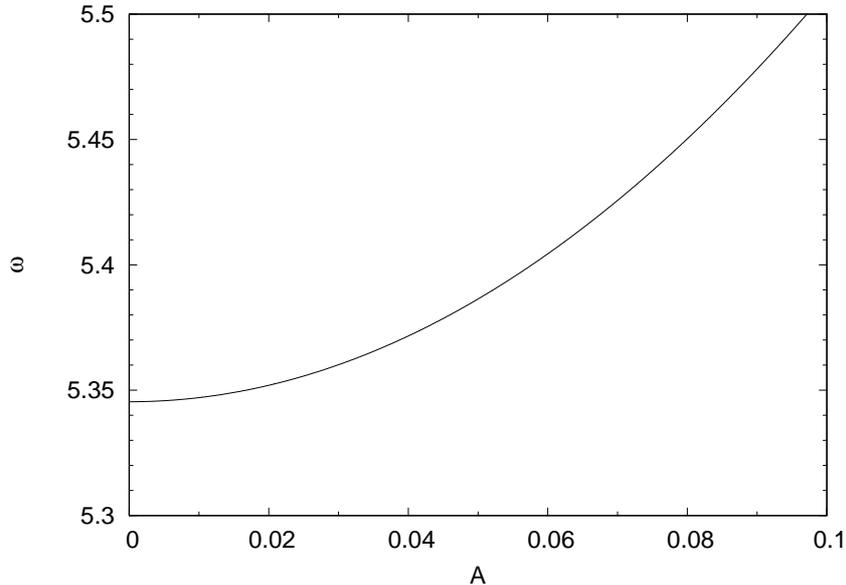}
  \caption{\label{fig2} The amplitude-frequency diagram $\omega(A)$ of the one-dimensional bush (mode~6) in the L-J chain with $N=11$ atoms.}
\end{figure}

The critical amplitude $A_\text{cr}$ of the stability loss (see Ref.~\cite{FPU1}) of mode~6 calculated with the aid of the Floquet method is $A_\text{cr}=0.0869$.

Three immobile atoms 3, 6, and 9 correspond to each of modes 4 and 8. The immobility of these atoms is the consequence of inversions localized at their positions. Both modes 4 and 8 possess the same symmetry, due to which the excitation of only one of them leads to the excitation of the second. The result is the \textit{two-dimensional} bush consisting of modes 4 and 8.

Two immobile atoms 4 and 8 correspond to modes 3 and 9, and these modes possess the same symmetry.

There is also mode~6 with zeros at the \textit{same atomic positions}, but it possesses a higher symmetry since it has additional zeros at atoms 2, 6, and 10. According to the bush theory, mode~6 must be involved into the vibrational process as a secondary mode of higher symmetry. As a result, the initial excitation of mode~3 (or mode~9) leads to the appearance of the \textit{three-dimensional} bush of modes 3, 6, 9 in the considered chain. During time evolution, energy wanders among the modes entering into the given bush, while the total energy of the initial excitation remains localized in this dynamical object.

Modes 2 and 10 possess one immobile atom position~6. When one of them is excited, the \textit{five-dimensional} vibrational bush appears, formed by all even modes (2, 4, 6, 8, and 10). In this case, the secondary modes 4, 6, and 8 possess higher symmetry: three immobile atoms in modes 4 and 8, and five such atoms in mode~6.

Excitation of one of the modes 1, 5, 7, 11 in which there are no immobile atoms (they possess only permutation~$\mathcal{P}$ located at the sixth atom) leads to the appearance of a \textit{general bush}, i.e., bush, which includes all the 11 normal modes.

A similar analysis of the selection rules for the excitation transfer between vibrational modes of different symmetry can be carried out for monoatomic chains with an arbitrary number of atoms. For example, for the cases $N=19$, $N=20$, $N=21$, the following results were obtained: see Tables~\ref{table10}, \ref{table11}, \ref{table12}.

\begin{table}
  \centering
  \begin{tabular}{cccc}
    \toprule
    Bush & Modes of       & Root        & Secondary\\
    dim. & the bush       & modes       & modes\\
    \midrule
    1    & 10             & 10          &  \\
    2    & 8,16           & 8,16        &  \\
    3    & 5,10,15        & 5,15        & 10 \\
    4    & 4,8,12,16      & 4,12        & 8,16 \\
    9    & All even modes & 2,6,14,18   & 4,8,10,12,16 \\
    19   & All modes      & 1,3,7,9,11, & All other modes \\
         &                & 13,17,19    & \\
    \bottomrule
  \end{tabular}
  \caption{\label{table10}Bushes for the chain with $N=19$ atoms.}
\end{table}

\begin{table}
  \centering
  \begin{tabular}{cccc}
    \toprule
    Bush & Modes of       & Root            & Secondary\\
    dim. & the bush       & modes           & modes\\
    \midrule
    1    & 14             & 14              &  \\
    2    & 7,14           & 7               & 14 \\
    3    & 6,12,18        & 6,12,18         &  \\
    6    & 3,6,9,12,15,18 & 3,9,15          & 6,12,18 \\
    10   & All even modes & 2,4,8,10,16,20  & 6,12,14,18 \\
    20   & All modes      & 1,5,11,13,17,19 & All other modes \\
    \bottomrule
  \end{tabular}
  \caption{\label{table11}Bushes for the chain with $N=20$ atoms.}
\end{table}

\begin{table}
  \centering
  \begin{tabular}{cccc}
    \toprule
    Bush & Modes of       & Root          & Secondary\\
    dim. & the bush       & modes         & modes\\
    \midrule
    1    & 11             & 11            &  \\
    5    & 4,8,12,16,20   & 4,8,12,16,20  &  \\
    10   & All even modes & 2,6,10,14,18  & 4,8,12,16,20 \\
    21   & All modes      & 1,3,5,7,9,13, & All other modes \\
         &                & 15,17,19,21   &  \\
    \bottomrule
  \end{tabular}
  \caption{\label{table12}Bushes for the chain with $N=21$ atoms.}
\end{table}

In the above tables, we indicate the dimension of each possible bush in the chain with 19, 20, and 21 atoms, all its modes, root modes (excitation of any of these mode leads to the appearance of the bush), and secondary modes that are involved into the vibrational process due to their interactions with root modes.

It might seem that bushes of modes are only related to atomic vibrations with large amplitudes. This is not the case; bushes are adequate dynamical objects for any values of the atomic amplitudes, in particular, for small ones.

So far, we have excited bushes by initial excitation of the one root mode. However, the same bush can be excited by assigning nonzero amplitudes to all its modes. As an example, in Fig.~\ref{fig100} we show the time evolution of the three-dimensional bush of modes 3, 6, 9 for the L-J chain with $N=11$ mobile atoms. This figure confirms the main thesis that any bush is a unified dynamical object in the sense that energy wanders only between its modes without transferring to other vibrational modes of the system. This thesis remains valid also for the case of small amplitudes.

\begin{figure}
  \centering
  \includegraphics[width=.7\linewidth]{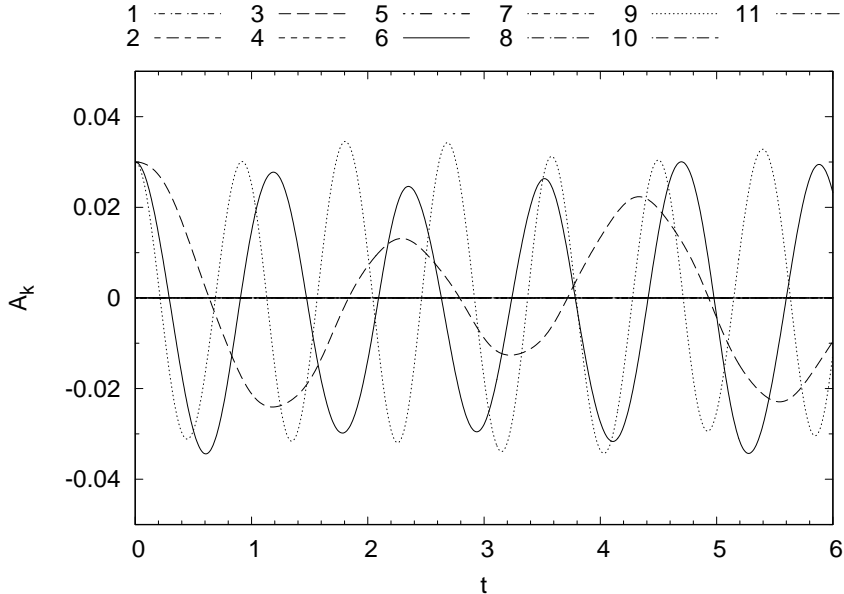}
  \caption{\label{fig100} Time evolution of the three-dimensional bush in the L-J chain with $N=11$ mobile atoms.}
\end{figure}

Note that bushes can be excited in the considered system not only by specifying the initial values of the amplitudes of their modes. They can be also excited by acting on the system with some external time-periodic forces at the resonant frequencies. Methods for the vibrational bushes excitation will be published elsewhere.

\section{\label{sec7}Study of the vibrational nonlinear modes in the framework of the density functional theory}

The previous part of the present work was devoted to the study of nonlinear dynamics of monoatomic chains, whose interatomic interactions are described by certain pair potentials. This approach corresponds to methods of molecular dynamics, in which individual atoms are replaced by material points interacting via some phenomenological potentials, and their motion is described by classical Newton's laws.

From a general theoretical standpoint, this is a rather rough approximation for studying nonlinear dynamics of atomic systems and, therefore, one can ask to what extent the authors' results, obtained by molecular dynamics methods, are valid within the framework of the density functional theory~\cite{Kohn1999} (DFT-methods)?

In this theory, the Born-Oppenheimer approximation is used to study the dynamics of electron-nuclear subsystems: the motion of light electrons is described in the framework of quantum mechanics by solving the Kohn-Sham integro-differential equations, while the motion of heavy nuclei is described by the classical Newton equations.

This computation procedure for studying the dynamics of the atomic system is as follows. The equilibrium state of the electron subsystem for the initially fixed positions of all the nuclei is found. Then one micro step in time is made for the nuclear subsystem using some numerical method for solving Newton's ordinary differential equations, as a result of which a new spatial distribution of the nuclei is obtained. The above computational procedure is repeated until the equilibrium state of the considered electron-nuclear system is reached. It is essential for this approach that each atom is not considered as a material point corresponding to its nucleus, but turns out to have an electron shell and these electronic shells are deformed when the distances between atoms change during their motion.

Since the mathematical apparatus of the density functional theory differs significantly from that of the molecular dynamics, there may be doubts about the adequacy of the theory of bushes of nonlinear normal modes in the framework of the DFT modeling. When discussing this issue, it should be taken into account that the bush theory is based on the group-theoretical methods, which use only the symmetry properties of the considered systems. This is a strong argument in favor of the correctness of this theory for a wide variety of DFT models, which was confirmed in our papers on nonlinear dynamics of different molecular and crystal structures~\cite{PRB2014, SF62015, CMS2017, Graphene2017, Diamond2018}.

In the present work, we use the density functional theory to consider monoatomic chains of carbon atoms, the so-called carbynes~\cite{Casari2016}. Let us explain this choice.

Monoatomic carbon chains can exist in two different forms. These are \textit{cumulene} with double bonds between all atoms [chemical structure \ch{(=C=C=)_{$n$}}] and \textit{polyyne} with the alternation of single and triple bonds [chemical structure \ch{(-C+C-)_{$n$}}]. As a result, all bond lengths are identical in cumulene, while polyyne demonstrates alternation of short and long bonds.

Carbynes possess many unique properties. In particular, they are the strongest material known at the present time, cumulene is a conductor better than linear gold chain and can be considered as the thinnest nanowire. An energy gap arises in the electron spectrum of the cumulene with an appropriate strain, and it changes from the conductor to a semiconductor or dielectric. Direct measurement of the electrical conductivity of the finite carbyne chains was performed in~\cite{Cretu2013}.

It is assumed that the unique carbyne properties will allow using this material in the future for various purposes of nanotechnology. In particular, carbyne chains are interesting because they can connect different fragments of graphene and can be used as the smallest nanowire in metal-matrix nanocomposites. Carbyne is considered as a promising material for spintronics, straintronics, as well as for hydrogen storage in hydrogen technology.

The synthesis of carbyne and the experimental study of its properties present great difficulties~\cite{Casari2016}. So far it was possible to synthesize free carbon chains consisting of only a few dozen atoms. Construction of such chains and studying their mechanical properties was carried out in a number of works using high-resolution transmission electron microscopy (HRTEM). This technology allows one to see and operate with individual carbon atoms, in particular, to determine such characteristics of carbyne chains as the bond lengths. The authors of the paper~\cite{Chuvilin2009} wrote ``The dynamics are observed in real time at atomic resolution with enough sensitivity to detect every carbon atom that remains stable for a sufficient amount of time''.

Free carbon chains of any length must be terminated by molecular complexes to ensure their stability. In principle, they can be obtained with two different terminations: $sp^3$ termination, resulting in carbon atoms linked by alternated single and triple bonds (polyyne) and, thus, with alternating bond lengths, and $sp^2$ termination resulting in double bonds (cumulene). Note that experimentally polyyne is found to be more stable than cumulene.

Since the experimental study of carbynes presents great difficulties, theoretical methods are very important in studying their properties. Most of these methods are based on the density functional theory (DFT)~\cite{HohenbergKohn1964, KohnSham1965, Kohn1999}, implemented in a number of powerful computational packages, such as ABINIT, Quantum Espresso, VASP, and others. Many interesting results were obtained in this way for infinite and finite carbon chains~\cite{Tongay2004, Cahangirov2010, Sorokin2011, Castelli2012, Liu2013, Artyukhov2014, Casillas2014, Timoshevskii2015, Torre2015, Santhibhushan2016, Liu2017, CMS2017}. In Ref.~\cite{Artyukhov2014}, for long-strained carbon chains with an even number of atoms, the Peierls phase transition was predicted above a certain threshold of the strain. As a result of this transition, carbyne transforms from metallic state to insulator state, and one can use this property in nanodevices to control the conductivity of the material by mechanical action.

Properties of sufficiently short carbon chains are discussed in~\cite{Tongay2004, Cahangirov2010, Sorokin2011, Casillas2014, Timoshevskii2015}. In Ref.~\cite{Cahangirov2010} DFT methods allowed detailed investigation of the so-called ``parity effect'' for the strained carbon chains of finite size. It means that the distribution of bond lengths and magnetic moments at atomic sites exhibit even-odd disparity depending on the number of carbon atoms in the chain. The dependence of bond lengths on the type of saturation of carbon chains at their both ends was also studied. Several hydrogen atoms can be attached to their ends to passivate chemically active ends of the chains. If two hydrogen atoms are attached to each end of the chain, then the bond lengths in its middle part correspond approximately to the cumulene structure. If only one hydrogen atom is attached to each end of the chain, then the polyyne structure appears. This problem, as well as the methods for identifying various forms of carbyne by infrared absorption spectra, are discussed in detail in Ref.~\cite{Casari2016}.

In Ref.~\cite{CMS2017}, large-amplitude atomic oscillations in the strained carbon chains were studied with the aid of DFT modeling, and a sharp softening of the $\pi$-mode frequency was found above a certain critical strain value. Condensation of this mode also leads to the Peierls transition discussed in~\cite{Artyukhov2014}. Moreover, the soft mode concept allowed the authors to suggest that there can exist two new forms of carbyne, which differ from the polyyne in the type of alternation of short and long chemical bonds. Actually, this conclusion is based on the theory of bushes of nonlinear normal modes in systems with space or point symmetries group~\cite{DAN1, PhysD98, ChechinZhukov2006} depending on the number of the particles of the chain.

To study the DFT models, the computational package ABINIT~\cite{ABINIT2020} has been used with optimized norm-conserving Vanderbilt pseudopotentials~\cite{Hamann2013, vanSetten2018} and the Perdew-Burke-Ernzerhof (PBE) exchange-correlation functional~\cite{PBE1996}. The kinetic energy cutoff was set to $42$~Hartree, the time step for ions to $20$~atomic units of time, and the tolerance on the difference of forces to $10^{-4}$~Hartree/Bohr.

As it was noted above, the stable state of carbyne is polyyne, which corresponds to the alternation of interatomic bond lengths. However, our analysis of the intermode interactions, carried out in the previous part of this work, corresponds to monoatomic chains with the \textit{same} bond lengths between all atoms, which corresponds to the cumulene state of carbon chains. In order for the system's equilibrium state to be cumulene, some atomic complexes creating $sp^2$ hybridization can be attached to both ends of the chain. We fulfill this condition by attaching oxygen atoms to the ends of the chain to compensate two free chemical bonds of the terminal carbon atoms. We consider these additional oxygen atoms, as well as the terminal carbon atoms, as immobile. Thus, the initial equilibrium state of the chain can be written in the form \ch{O=C=C=C=...=C=C=C=O}.

The excitation of any vibrational bush in such a chain is carried out by setting the pattern of atomic displacements, corresponding to the given root mode with a certain amplitude~$A$, and zero velocities of all atoms as the initial conditions for the DFT descent method.

Let us consider the obtained results.

In Fig.~\ref{fig50}, we present the evolution of the one-dimensional bush generated by the root mode~6 for the carbon chain of $N=11$ mobile atoms.

\begin{figure}
  \centering
  \includegraphics[width=.7\linewidth]{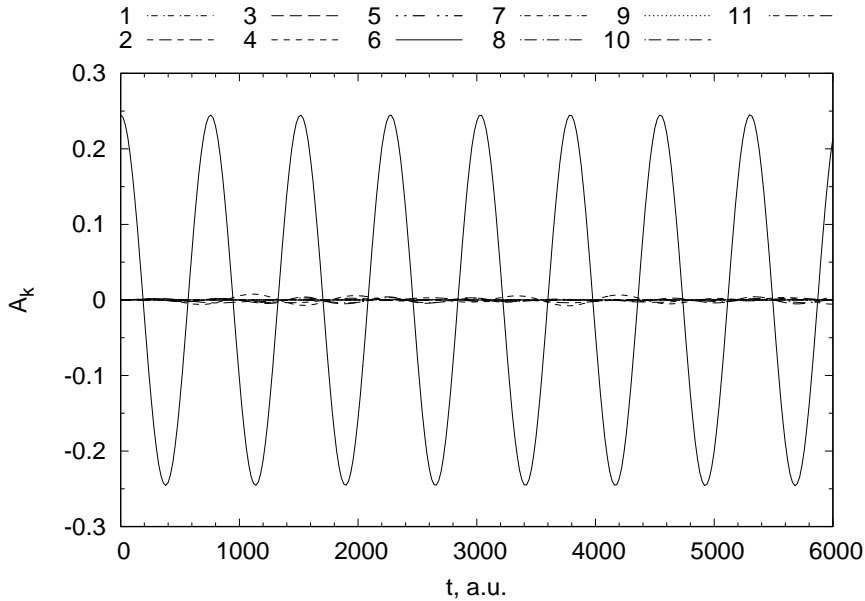}
  \caption{\label{fig50} One-dimensional quasi-bush (mode~6) in the carbon chain of $N=11$ mobile atoms.}
\end{figure}

According to the analysis performed in the previous part of the present work, excitation of this mode should not lead to excitation of other modes since mode~6 has the highest symmetry among all vibrational modes for the considered chain.

However, it can be seen from Fig.~\ref{fig50} that mode~6 exists against the background of other modes that oscillate with very small amplitudes in relation to its own amplitude. This phenomenon is associated with the peculiarities of the boundary conditions that we used in our DFT-modeling.

Pronounced oscillations of the mode~6 against the small background of other modes show that bushes can possess a sufficient degree of stability with respect to different small perturbations.

Vibrational background that is shown in Fig.~\ref{fig50} corresponds to the thermal noise with $T=1.1$~K.

Strictly speaking, we have in Fig.~\ref{fig50} not an exact bush, but some one-dimensional \textit{quasi-bush}. Oscillations of the mode~6 are slightly modulated by a vibrational background of other modes, but this does not destroy the discussed quasi-bush.


The Lennard-Jones model provides a very good approximation to the amplitude-frequency diagram obtained in the framework of the density functional theory. This fact is evident from Fig.~\ref{fig51} where we show points of the amplitude-frequency diagram for mode~6 obtained using DFT calculations for the carbon chain with $N=11$ mobile atoms, which, after the appropriate scaling, were superimposed on the curve of the amplitude-frequency diagram in dimensionless units for the corresponding L-J chain (it was already shown in Fig.~\ref{fig2}). Dimensionless units are shown against the top horizontal and the right vertical axes. Dimensional data obtained as the result of DFT calculations are plotted against the bottom horizontal and the left vertical axes.

\begin{figure}
  \centering
  \includegraphics[width=.7\linewidth]{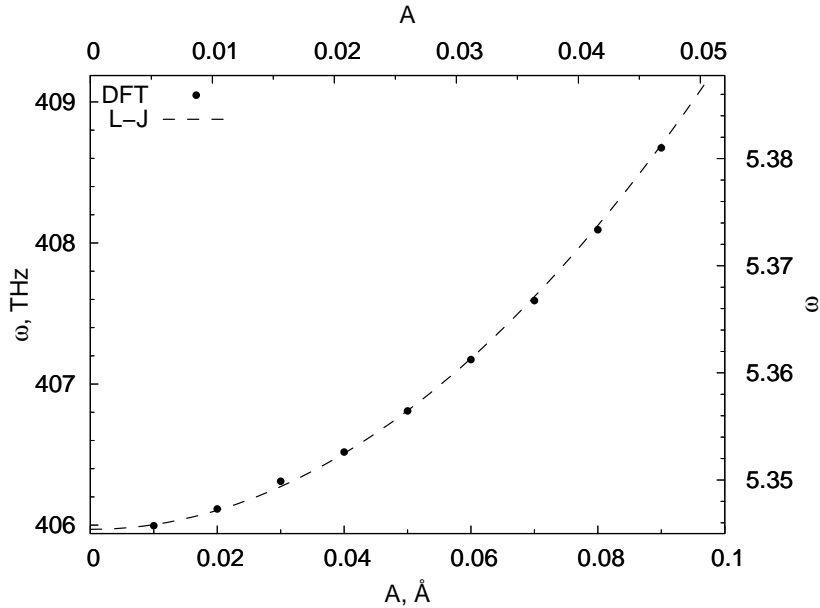}
  \caption{\label{fig51} The amplitude-frequency diagram $\omega(A)$ of the one-dimensional bush (mode~6) in the carbon chain with $N=11$ atoms from DFT-modelling (points; bottom and left axes) in comparison with corresponding diagram of the same bush for the L-J model (dashed line; top and right axes).}
\end{figure}

In Fig.~\ref{fig52}, we demonstrate the time evolution of the two-dimensional quasi-bush formed from modes 4 and 8 for the carbyne chain of $N=11$ mobile atoms. At the initial instant, only the root mode~4 was excited, which involves the secondary mode~8 into the vibrational process. The latter mode has the same symmetry as mode~4 and, therefore, may itself be the root mode of the same bush.

\begin{figure}
  \centering
  \includegraphics[width=.7\linewidth]{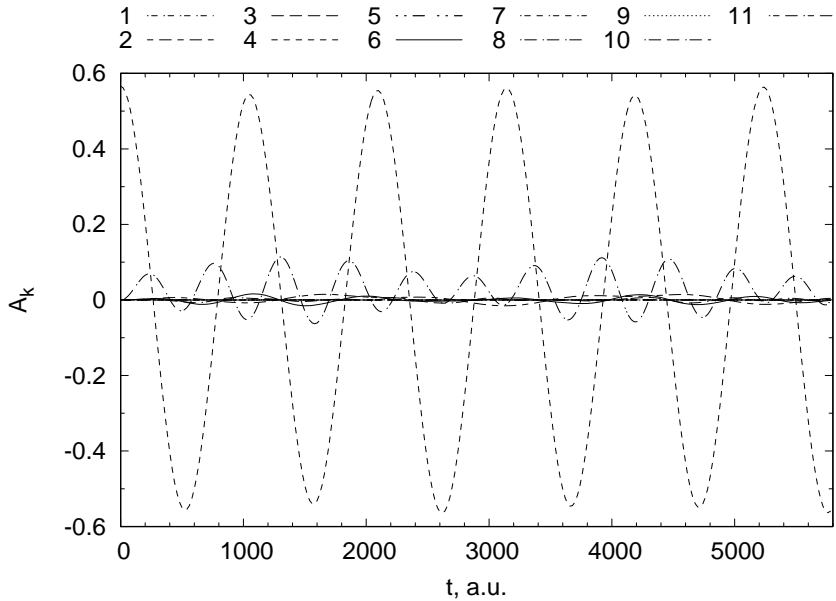}
  \caption{\label{fig52} Time evolution of two-dimensional quasi-bush (modes 4 and 8) in the carbon chain with $N=11$ mobile atoms.}
\end{figure}

It can be seen from Fig.~\ref{fig52} that the secondary mode has a smaller amplitude compared to the root mode. This is due to the way in which this bush was excited.

It should be noted that our two-dimensional bush, as well as any other bush as a single dynamical object, can be excited by setting arbitrary amplitudes of all its modes. Certainly, the value of these amplitudes and, therefore, the energy of the initial excitation cannot be arbitrarily large---they should not lead to the loss of stability of the considered bush~\cite{FPU1}.

In conclusion, we would like to note that this work was not aimed at a detailed study of the bushes of nonlinear vibrational modes in carbynes. Our goal was to show that bushes of the considered type for monoatomic chains with fixed boundary conditions can really exist in the framework of the DFT approach.

\section{\label{sec8}Summary}

The traditional approach to the study of nonlinear dynamics of atomic systems starts with the harmonic approximation, after which the influence of different types of weak anharmonisms on the considered phenomenon is investigated. It was in this direction that the classic work~\cite{FPU1955} of Fermi, Pasta, and Ulam (FPU) was carried out, which actually led to the beginning of rapid development of the modern nonlinear dynamics. Within the framework of the above approach, the analysis of intermode interactions in the one-dimensional FPU chains allowed Zabusky and Kruskal to introduce the concept of solitons~\cite{ZabuskyKruskal1965}.

However, in these and many other works devoted to the dynamics of the FPU chains, the \textit{symmetry aspects} of intermode interactions were not investigated.

On the other hand, in our works~\cite{FPU1, FPU2}, \textit{bushes} of nonlinear vibrational modes of different symmetry were found for the case of the FPU chains with periodic boundary conditions. These results were obtained with the aid of the general theory of bushes of nonlinear normal modes developed in~\cite{DAN1, PhysD98}.

In the geometric sense, each bush defines an invariant manifold of the considered system on which some exact solution of the original nonlinear problem lies. In the dynamical sense, every bush represents a set of nonlinear modes that do not change over time, but the amplitudes of these modes change since they exchange energy with each other. The full bush energy is conserved because it is not transferred to the modes that don't belong to the given bush.

In this paper, we present a detailed consideration of the bushes of the vibrational modes in the monoatomic chains with fixed ends. In particular, we have proved some theorems that form the foundation for the methods of their construction.

Selection rules for energy exchange between modes of different symmetry are found. It is proved that these selection rules, based on studies of normal modes, i.e., for the atomic oscillations with small amplitudes, are also valid for the case of \textit{large amplitudes}. Being obtained only from the analysis of the considered systems' symmetry, these selection rules can be applied to the chains with any type of interatomic interactions. We have demonstrated this conclusion for the Fermi-Pasta-Ulam chains and for the Lennard-Jones chains. It is important to emphasize that discussed selection rules are valid for vibrations with \textit{arbitrary} amplitudes, and bushes of nonlinear normal modes form a class of \textit{exact} solutions beyond any perturbation theory.

Our attempt to construct bushes of nonlinear modes for monoatomic chains with the fixed ends within the framework of the density functional theory led to the discovery of dynamical objects that can be called \textit{quasi-bushes}. Each quasi-bush differs from the exact one in that it exists against the background of \textit{small} vibrations of other modes, which don't belong to the given bush. This background can be considered as an analogue to a low-temperature thermal noise. A more detailed study of the bushes of nonlinear modes with the aid of the density functional theory will be published elsewhere.

\section*{Acknowledgments}

The study was financially supported by the Ministry of Science and Higher Education of the Russian Federation [State task in the field of scientific activity, scientific project No.~0852-2020-0032 (BAS0110/20-3-08IF)]. We are sincerely grateful to N.~V.~Ter-Oganessian for useful discussions.

\bibliography{physd2021}

\end{document}